\documentclass[10pt,journal]{IEEEtran}
\usepackage{amsfonts,amsmath,amssymb}
\usepackage{cite}
\usepackage{graphicx}
\usepackage{url}
\newtheorem{theorem}{Theorem}
 \newtheorem{lemma}{Lemma}
  \newtheorem{proposition}{Proposition}

  \newtheorem{remark}{Remark}

\begin{document}
%\baselineskip 4.8ex

% paper title
\title{  %\Large\bf
Full Diversity  Space-Time Block Codes with Low-Complexity Partial Interference Cancellation Group Decoding}
 % \author{Full version of ISIT paper.}
\author{Wei~Zhang,~\IEEEmembership{Member,~IEEE,} ~Long~Shi,~\IEEEmembership{Student Member,~IEEE,}
        and ~Xiang-Gen~Xia,~\IEEEmembership{Fellow,~IEEE}
  %  \thanks{Manuscript revised on 22 December, 2009.}

 \thanks{W.  Zhang and L. Shi are with School of Electrical Engineering and Telecommunications,   The University of New South Wales,
  Sydney, Australia (e-mail: \{w.zhang; long.shi\}@unsw.edu.au). Their work was supported in part by the Australian Research Council Discovery Project DP1094194.}%
   \thanks{X.-G. Xia is with Department of Electrical and Computer Engineering, University of Delaware, DE 19716, USA (e-mail: xxia@ee.udel.edu). His work was supported in part by the Air Force Office of Scientific
Research (AFOSR) under Grant No. FA9550-08-1-0219.
}
}

\maketitle

\begin{abstract}
Partial interference cancellation (PIC) group decoding proposed by Guo and Xia is an attractive low-complexity alternative to the optimal processing for multiple-input multiple-output (MIMO) wireless communications. It can well deal with the tradeoff among rate, diversity and complexity of space-time block codes (STBC). In this paper,  a systematic design of full-diversity STBC with low-complexity PIC group decoding is proposed. The proposed code design is featured as a group-orthogonal STBC by replacing every element of an Alamouti code matrix with an elementary matrix composed of multiple diagonal layers of coded symbols. With the PIC group decoding and a particular grouping scheme, the proposed STBC can achieve  full diversity, a rate of $(2M)/(M+2)$ and a low-complexity decoding for $M$ transmit antennas. Simulation results show that the proposed codes can achieve the full diversity with PIC group decoding while requiring half decoding complexity of the existing codes.

\end{abstract}


\section{Introduction}
\label{sec:intro}

Multiple-input multiple-out (MIMO) wireless communications have been witnessed to offer large gains in spectral efficiency and reliability \cite{Telatar,Naguib}. Efficient designs of signal transmission schemes include space-time (ST) codes over MIMO systems have been active areas of research over the past decade \cite{ZhangMag}. Orthogonal ST block code (OSTBC) is one of the most powerful ST code designs due to its simple low-complexity maximum-likelihood (ML) decoding  while achieving maximum diversity gain \cite{Alamouti, Tarokh98, Tarokh00, Lu}. However,  it is found that OSTBC has a low code rate that cannot be above $3/4$ symbols per channel use for more than two transmit antennas \cite{Wang}. To improve the code rate of the STBC, numerous code designs have been developed including quasi-orthogonal STBC  \cite{Jafar,Tirkkonen, Papadias,Tirkkonen2,Sharma,SuXia,Khan,Yuen,HWang,Dao,Karmakar} and STBC based on algebraic number theory  \cite{Damen,gamal,rajan0, WangGY, WangGY2, rajan1, kumar1, kumar2, oggier,oggier1,Guo09}.
Two typical designs of those codes are threaded algebraic ST (TAST) codes \cite{gamal, WangGY2} and cyclic division algebra based ST codes \cite{rajan0, rajan1, kumar1, kumar2, oggier,oggier1,Guo09} which
 have been shown to obtain full rate and full diversity. The full rate means $M$ symbols per channel use for $M$ transmit antennas. Note that the OSTBC for two transmit antennas, also namely Alamouti code \cite{Alamouti}, has a rate of $1$ symbol per channel use only, i.e., two independent information symbols are sent through a codeword occupying two symbol intervals. Since most of the high-rate STBC are designed based on the rank criterion which was derived from the pairwise error probability of the ST codes with ML decoding \cite{Tarokh98}, they have to rely on the ML decoding to collect the full diversity. Considering that the ML decoding complexity grows exponentially with the number of information symbols embedded in the codeword, the high-rate STBC obtain the full diversity at a price of the large decoding complexity.

Recently, several fast decodable STBC have been proposed to reduce the high decoding complexity while not compromising too much performance gains \cite{Choi, Hong, Calderbank}.  MIMO systems with linear receivers have also received a lot of research attention and information-theoretic analysis has been done in \cite{Aria,Caire,Tse,Aria2,Jiang}. Efficient   designs of ST codes for transmission over MIMO systems with linear receivers have also been studied in \cite{Liu,Shang,ZhangICASSP09,ZhangISIT09,WangHM}. Linear receiver based STBC designs are attractive because they can exploit both gains of efficiency and reliability of the signal transmission over MIMO systems with a low-complexity receiver such as zero-forcing (ZF) or minimum mean square error (MMSE) receiver. Similar to the OSTBC, the STBC designs in \cite{Liu,Shang,ZhangICASSP09,ZhangISIT09,WangHM} can also obtain full diversity with linear receivers. However, it is found that the rate of the linear receiver  based STBC is upper bounded by one \cite{Shang}, though it is larger than that of OSTBC.   To address the complexity and rate tradeoff, a partial interference cancelation (PIC) group decoding for MIMO systems was proposed and the design criterion of STBC with PIC group decoding was also derived in \cite{Guo}. In fact, the PIC group decoding can be viewed as an intermediate decoding approach between the ML receiver and the ZF receiver by trading a simple single-symbol decoding complexity for a higher code rate more than one symbol per channel use. Very recently, a systematic design of STBC achieving full diversity with PIC group decoding was  proposed in \cite{ZhangGC09}. However, the decoding complexity of the STBC design in \cite{ZhangGC09} is still equivalent to a joint ML decoding  of $M$ symbols.

In order to further reduce the decoding complexity, in this paper we propose a new design of STBC with PIC group decoding which can obtain both full diversity and low-complexity decoding, i.e., only half complexity of the STBC in \cite{ZhangGC09}.  Our proposed STBC is featured as an Alamouti block matrix, i.e., every element of the $2\times 2$ Alamouti code matrix is replaced by an elementary matrix and each elementary matrix is designed from multiple diagonal layers. It should be mentioned that in \cite{WangHM} the similar Alamouti block matrix was used where each entry of the Alamouti matrix was replaced by a Toeplitz STBC. The major difference between the STBC in \cite{WangHM} and our proposed STBC lie in the construction of elementary matrix, i.e., the Toeplitz matrix used in \cite{WangHM} and the multiple diagonal layers used in our codes. While the STBC in \cite{WangHM} achieves the full diversity with linear receivers but the code rate is not more than $1$.
It will be shown that our proposed STBC can achieve full diversity under both ML and PIC group decoding and the code rate can be up to $2$ when full diversity is obtained. Our simulation results demonstrate that the codes can obtain similar good performance to the codes in \cite{ZhangGC09} but   a half decoding complexity is reduced.
This paper is organized as follows. In Section II, the system model is described and the PIC group decoding algorithm is reviewed. In Section III, a design of STBC achieving full diversity with a reduced-complexity PIC group decoding is proposed. The full diversity is proved when PIC group decoding is applied. In Section IV, a few code design examples are given. In Section V, simulation results are presented.  Finally, we conclude the paper in Section VI.

 \emph{Notation:} Throughout this paper we use the following notations.
  Column vectors (matrices) are denoted by boldface lower
  (upper) case letters. Superscripts $^t$ and $^H$ stand for transpose and
conjugate transpose, respectively.  $\mathbb{C}$ denotes the field of complex numbers.
$\mathbf{I}_n$ denotes the $n\times n$ identity matrix, and  $\mathbf{0}_{m\times n}$ denotes the $m\times n$ matrix  whose elements are all $0$. %$\mathrm{vec}(\mathbf{X})$ is the vectorization of matrix $\mathbf{X}$ by stacking the columns of $\mathbf{X}$ on top each other.
$\det(\mathbf{X})$ represents the determinant of the matrix $\mathbf{X}$. $\otimes$ denotes the Kronecker product. $||\mathbf{X}||$ denotes the Frobenius norm of matrix (vector) $\mathbf{X}$.

\section{System Model}
\label{sec:system}

We consider a MIMO system with $M$ transmit antennas and $N$ receive antennas where data symbols
$\{S_l\}, l=1, \cdots, L$, are sent to the receiver over block fading channels. Before the data transmission,   the information symbol vector $\mathbf{s}=(S_1,\cdots, S_{L})^t$, selected from a signal constellation $\mathcal{A}$ such as QAM,  are encoded into a space-time block codeword matrix $\mathbf{X}(\mathbf{s})$ of size $T\times M$, where $T$ is the block length of the codeword. For any  $1\leq t\leq T$ and $1\leq m\leq M$, the $(t,m)$-th entry   of $\mathbf{X}(\mathbf{s})$ is  transmitted to the receiver from the $m$-th antenna during the $t$-th symbol period through flat fading channels. The received space-time signal at $N$ receive antennas, denoted by the $T\times N$ matrix $\mathbf{Y}$, can be expressed as
 \begin{eqnarray}\label{eqn:Y}
 \mathbf{Y}=\sqrt{\frac{\rho}{\mu}}\mathbf{X}(\mathbf{s})\mathbf{H}+\mathbf{W},
 \end{eqnarray}
where $\mathbf{W}$ is the noise matrix of size $T\times N$ whose elements are of i.i.d. with circularly symmetric complex Gaussian distribution with zero mean and unit variance denoted by  $\mathcal{CN}(0,1)$, $\mathbf{H}$ is the $M\times N$ channel matrix whose entries are also i.i.d. with distribution $\mathcal{CN}(0,1)$, $\rho$ denotes the average signal-to-noise ratio (SNR) per receive antenna and $\mu$ is the normalization factor such that the average  energy of the coded symbols transmitting from all antennas during one symbol period is one. We suppose that  channel state information (CSI) is known at receiver only. Therefore, the signal power is allocated uniformly across the transmit antennas in the absence of transmitter CSI.

In this paper, we consider that  information symbols $\{S_l\}, l=1, \cdots, L$ are coded by linear dispersion STBC as $\mathbf{X}(\mathbf{s})$.
To decode the transmitted sequence $\mathbf{s}$, we need to extract $\mathbf{s}$ from $\mathbf{X}(\mathbf{s})$. Through some operations, we can get an equivalent signal model from (\ref{eqn:Y}) as \cite{Liu}\cite{Shang}\cite{Guo}
\begin{eqnarray}\label{eqn:Y2}
\mathbf{y}=\sqrt{\frac{\rho}{\mu}}\mathcal{H}\mathbf{s}+\mathbf{w},
 \end{eqnarray}
where $\mathbf{y}$ is a  vector of length $TN$, $\mathbf{w}$ is a  $TN\times 1$ noise vector, and $\mathcal{ {H}}$ is an equivalent channel matrix of size $TN\times L$ with $L$ column vectors $\{\mathbf{g}_{l}\}$ for $l=1,2,\cdots, L$, i.e.,  $\mathcal{ {H}}=\left[\begin{array}{cccc} \mathbf{g}_{1}& \mathbf{g}_{2}& \cdots & \mathbf{g}_{L}\end{array}\right]$.

For an ML receiver, the estimate of $\mathbf{\hat{s}}^{\mathrm{ML}}$   that achieves the minimum of the squared Frobenius norm is given by
 \begin{eqnarray}\label{eqn:ML}
\mathbf{\hat{s}}^{\mathrm{ML}}=\arg\min_{\mathbf{s}\in\mathcal{A}^L}\left\|\mathbf{y}-\sqrt{\frac{\rho}{\mu}}\mathcal{H}\mathbf{s}\right\|^2 .
   \end{eqnarray}
It is  known that the ML decoding has prohibitively large computational complexity because it requires an exhaustive search over all candidate vectors $\mathbf{s}$. Conventional linear receivers such as ZF or MMSE detection can reduce the decoding complexity to single symbol decoding, but it may lose some performance benefits such as diversity gain. Although some linear receiver based STBC have been recently proposed in \cite{Liu,Shang,ZhangICASSP09,ZhangISIT09,WangHM}, they suffer the rate loss and the symbol rate cannot be above $1$ symbol per channel use. Very recently, in \cite{Guo}, a PIC group decoding was proposed to deal with the tradeoff among rate, diversity and complexity. In fact, the PIC group decoding can be viewed as a flexible decoding algorithm with adjustable receiver structure from a linear receiver to  an ML receiver. Next, we shall introduce the PIC group decoding proposed in \cite{Guo}.

\subsection{PIC Group Decoding \cite{Guo}}

Define index set $\mathcal{I}$ as
$
\mathcal{I}=\{1, 2, \cdots, L\},
$
where $L$ is the number of information symbols in $\mathbf{s}$. We then partition $\mathcal{I}$ into $P$ groups: $\mathcal{I}_1, \mathcal{I}_2,\cdots, \mathcal{I}_P$ with
$
\mathcal{I}_p=\{I_{p,1}, I_{p,2},\cdots, I_{p,l_p}\}, \,\,\, p=1,2,\cdots, P,
$
where $l_p$ is the cardinality of the subset $\mathcal{I}_p$. We call $\mathcal{I}=\{\mathcal{I}_1, \mathcal{I}_2,\cdots, \mathcal{I}_P\}$ a grouping scheme. For such a grouping scheme, we have
$
\mathcal{I}=\bigcup_{p=1}^{P}\mathcal{I}_p$ and $\sum_{p=1}^Pl_p=L.
$

Define
$
\mathbf{s}_{p}=\left[\begin{array}{cccc}
                                  S_{I_{p,1}}   & S_{I_{p,2}}  &  \cdots & S_{I_{p,l_p}}
                                 \end{array}
\right]^t$ and $\mathbf{G}_{p}=\left[\begin{array}{cccc}
                                  \mathbf{g}_{I_{p,1}}   & \mathbf{g}_{I_{p,2}}  &  \cdots & \mathbf{g}_{I_{p,l_p}}
                                 \end{array}
\right]$, for $p=1, \cdots, P$.
With these notations, (\ref{eqn:Y2}) can be written as
 \begin{eqnarray}
 \mathbf{y}=\sqrt{\frac{\rho}{\mu}}\sum_{p=1}^P\mathbf{G}_{p}\mathbf{s}_p+\mathbf{w}.
  \end{eqnarray}
Suppose that we want to decode the symbols embedded in the group $\mathbf{s}_p$. The  PIC group decoding first implements linear interference cancellation with a suitable choice of matrix $\mathbf{Q}_p$ in order to completely eliminate the interferences from other groups \cite{Guo}, i.e., $\mathbf{Q}_p\mathbf{G}_q=\mathbf{0}$, $\forall q\neq p$ and $q=1,2,\cdots, P$. Then, we have
\begin{eqnarray}\label{eqn:GZF}
\mathbf{z}_p&\triangleq &\mathbf{Q}_p\mathbf{y}\nonumber\\
&=&\sqrt{\frac{\rho}{\mu}} \mathbf{Q}_p\mathbf{G}_{p}\mathbf{s}_p+\mathbf{Q}_p\mathbf{w},\,\,\,\,\, p=1,2,\cdots, P,
 \end{eqnarray}
where the interference cancellation matrix $\mathbf{Q}_p$ can be chosen as follows \cite{Guo},
\begin{eqnarray}\label{eqn:Qp}
\mathbf{Q}_p=\mathbf{I}-\mathbf{G}_p^c \left(\left(\mathbf{G}_p^c\right)^H\mathbf{G}_p^c\right)^{-1}\left(\mathbf{G}_p^c\right)^H,\, p=1,\cdots, P,
 \end{eqnarray}
 when the following matrix has full column rank:
 \begin{eqnarray}
 \mathbf{G}_p^c=\left[\begin{array}{cccccc}
                       \mathbf{G}_1 & \cdots & \mathbf{G}_{p-1} & \mathbf{G}_{p+1} & \cdots & \mathbf{G}_P
                     \end{array}
\right].
 \end{eqnarray}
 If  $\mathbf{G}_p^c$ does not have
full column rank, then we need to pick a maximal linear independent vector
group from $\mathbf{G}_p^c$ and in this case a projection matrix $\mathbf{Q}_p$ can be
found too \cite{Guo}.
Afterwards, the symbols  in the group $\mathbf{s}_p$ are decoded with the ML decoding algorithm as follows,
 \begin{eqnarray}\label{eqn:PIC}
 \mathbf{\hat{s}}_p=\arg\min_{\mathbf{s}_p\in\mathcal{A}^{l_p}}\left\|\mathbf{z}_p-\sqrt{\frac{\rho}{\mu}}\mathbf{Q}_p  \mathbf{G}_p\mathbf{s}_p\right\|^2, \,\,\,p=1,2,\cdots,P.
    \end{eqnarray}

%Compared with ML decoding in  (\ref{eqn:ML}), the PIC group decoding is a reduced-complexity ML decoding with an intermediate complexity between ML detection and ZF %detection.

\begin{remark}[PIC Group Decoding Complexity] For the PIC group decoding, the following two steps are needed: the group zero-forcing to cancel the interferences coming from all the other groups as shown in (\ref{eqn:GZF}) and the  group ML decoding to jointly decode the symbols in one group  as shown in (\ref{eqn:PIC}). Therefore, the decoding complexity of the PIC group decoding should reside in the above two steps. Note that   the interference cancellation process shown in (\ref{eqn:GZF})  mainly involves with linear matrix computations, whose computational complexity  is small compared to the ML decoding for an exhaustive search of all candidate symbols. Therefore, to evaluate the decoding complexity of the PIC group decoding, we mainly focus on the computational complexity of the ML decoding within the PIC group decoding algorithm. The ML decoding complexity can be characterized by the number of Frobenius norms calculated in the decoding process \cite{ZhangGC09}. In the PIC group decoding algorithm the complexity is then   $\mathcal{O}= \sum_{p=1}^P |\mathcal{A}|^{l_p}$. It can be seen that the PIC group decoding provides a flexible decoding complexity which can be from the ZF decoding complexity $L|\mathcal{A}|$ to the ML decoding complexity $|\mathcal{A}|^{L}$.
\end{remark}

\begin{remark}[PIC-SIC Group Decoding] In \cite{Guo}, a successive interference cancellation (SIC)-aided PIC group decoding algorithm, namely PIC-SIC group decoding was proposed. Similar to the  BLAST detection algorithm \cite{Foschini}, the PIC-SIC group decoding is performed  after removing the already-decoded symbol set from the received signals to reduce the interference. If each group has only one symbol, then the PIC-SIC group decoding will be equivalent to the BLAST detection.

\end{remark}

In \cite{Guo},  full-diversity STBC design criteria were derived
when the PIC group decoding and the PIC-SIC group decoding are used at the receiver. In the following, we cite the main results of the STBC design criteria proposed in \cite{Guo}.
\begin{proposition}\label{prop1}
\cite[Theorem 1]{Guo} [\emph{Full-Diversity Criterion under PIC Group Decoding}]

For an  STBC  $\mathbf{X}$  with the PIC group decoding, the full diversity is achieved when
  \begin{enumerate}
    \item   the code $\mathbf{X}$ satisfies the full rank criterion, i.e.,
it achieves full diversity when
the ML receiver is used; \emph{and}
    \item
    for any $p$, $1\leq p\leq P$, any non-zero linear combination over $\Delta \mathcal{A}$ of the vectors in the $p$th group $\mathbf{G}_p$ does not belong to the space linearly spanned by all the vectors in the remaining vector groups, as long as $\mathbf{H}\neq 0$, i.e.,
    \begin{eqnarray}
\sum_{\forall i\in \mathcal{I}_p} a_i \mathbf{g}_{i} \neq \sum_{\forall j\notin \mathcal{I}_p} c_j \mathbf{g}_{j}, \,\, \,\,\,\,a_i\in \Delta \mathcal{A},\mathrm{not\, all\, zero}, \,\,c_j\in \mathbb{C}
\end{eqnarray}
where $
\mathcal{I}_p=\{I_{p,1}, I_{p,2},\cdots, I_{p,l_p}\}$ is the index  set corresponding to the vector group $\mathbf{G}_{p}$ and $\Delta \mathcal{A}=\{S-\hat{S}, | S, \hat{S}\in \mathcal{A}\}$.

  \end{enumerate}
\end{proposition}

\begin{proposition}\label{prop2}
\cite{Guo} [\emph{Full-Diversity Criterion under PIC-SIC Group Decoding}]

For an  STBC  $\mathbf{X}$  with the PIC-SIC group decoding, the full diversity is achieved when
  \begin{enumerate}
    \item   the code $\mathbf{X}$ satisfies the full rank criterion, i.e.,
it achieves full diversity when
the ML receiver is used; \emph{and}
    \item   at each decoding stage, for $\mathbf{G}_{q_1}$, which corresponds to the current to-be decoded symbol group $\mathbf{s}_{q_1}$,   any non-zero linear combination over $\Delta \mathcal{A}$ of the vectors in   $\mathbf{G}_{q_1}$ does not belong to the space linearly spanned by all the vectors in the remaining groups $\mathbf{G}_{q_2},  \cdots, \mathbf{G}_{q_L}$  corresponding to yet uncoded symbol groups, as long as $\mathbf{H}\neq 0$.
  \end{enumerate}
\end{proposition}

\section{Proposed STBC with PIC Group Decoding}\label{sec:new}
In this section, we propose a design of STBC which can achieve full diversity with a low-complexity PIC group decoding.
 Compared to the one proposed in \cite{ZhangGC09} whose PIC group decoding   consists of $P$ groups with a  joint ML decoding of $M$ symbols per group, in the following our new STBC with PIC group decoding has $2P$ groups with a joint ML decoding of $M/2$ symbols  per group.

\subsection{Code Design}

Our proposed space-time code $\mathbf{\mathbf{B}}$,  i.e., $\mathbf{X}(\mathbf{s})$ in (\ref{eqn:Y}), is of size $T\times M$ (for any given $T$, $M=2m$, $m$ is an integer,  and $T\geq M$) given by
\begin{eqnarray}\label{eqn:new}
%\begin{split}
 \mathbf{\mathbf{B}}_{M,T,P} =
 \left[\begin{array}{cc}
    \mathbf{C}^1_{\frac{M}{2},\frac{T}{2},P} & \mathbf{C}^2_{\frac{M}{2},\frac{T}{2},P} \\
    -(\mathbf{C}^2_{\frac{M}{2},\frac{T}{2},P})^\ast & (\mathbf{C}^1_{\frac{M}{2},\frac{T}{2},P})^\ast\\
  \end{array}
   \right],
\end{eqnarray}
where $P=\frac{T-M}{2}+1$ and the matrices  $\mathbf{C}^i_{\frac{M}{2},\frac{T}{2},P}$ ($i=1,2$)  of size $\frac{T}{2}\times \frac{M}{2}$ is given by
\begin{eqnarray}\label{eqn:wei}
% \mathbf{C}^i_{\frac{M}{2},\frac{T}{2},P}&=&
  \left[\begin{array}{cccc}
                     X_{(i-1)P+1,1} & 0  & \cdots  &  0 \\
                     X_{(i-1)P+2,1}& X_{(i-1)P+1,2} & \ddots  & \vdots  \\
                      \vdots & X_{(i-1)P+2,2} &  \ddots &  0 \\
                     X_{iP,1} & \vdots & \ddots & X_{(i-1)P+1,\frac{M}{2}}  \\
                     0 & X_{iP,2} & \ddots  & X_{(i-1)P+2,\frac{M}{2}} \\
                     \vdots & 0  & \ddots  & \vdots \\
                     0 &  \vdots &  \ddots & X_{iP,\frac{M}{2}}
                  \end{array}
 \right]
 \end{eqnarray} for $i=1,2$
with the $p$th diagonal layer from   left to  right written as the  vector $\mathbf{X}^i_p$ of length $M/2$, shown as
\begin{eqnarray}
  \mathbf{X}^i_p =\left[\begin{array}{ccc}
                      X_{(i-1)P+p,1}     & \cdots &  X_{(i-1)P+p,\frac{M}{2}}
                    \end{array}
  \right]^t
   \end{eqnarray} for $i=1,2$ and $p=1,\cdots, P$.
   The vector $\mathbf{X}^i_p$ is further given by
 \begin{eqnarray}\label{eqn:RM}
  \mathbf{X}^i_p =\mathbf{\Theta}\mathbf{s}^i_p, \,\,\,\,\,p=1,2,\cdots, P
   \end{eqnarray}
where $\mathbf{\Theta}$ is a $\frac{M}{2}\times \frac{M}{2}$ rotation matrix   and $\mathbf{s}^i_p$ is a length-$\frac{M}{2}$ vector of information symbols given by
   \begin{eqnarray}
  \mathbf{s}^i_p =\left[\begin{array}{ccc}
                     S_{1+q_{i,p}}    & \cdots &  S_{\frac{M}{2}+q_{i,p}}
                    \end{array}
  \right]^t,
   \end{eqnarray}
   with $q_{i,p}=(i-1)P\frac{M}{2}+(p-1)\frac{M}{2}$,
   for $i=1,2$ and $p=1,\cdots, P$.

    One code example for $4$ transmit antennas and $T=6$ is given by
 \begin{eqnarray}\label{eqn:Tx4a}
 \mathbf{\mathbf{B}}_{4,6,2} = \left[\begin{array}{cccc}
                      {X}_{1,1} & 0  &  {X}_{3,1}  &  0 \\
                      {X}_{2,1} &  {X}_{1,2} &  {X}_{4,1}&  {X}_{3,2} \\
                      0 &  {X}_{2,2} & 0 &  {X}_{4,2} \\
                     -{ {X}}^\ast_{3,1} & 0 & { {X}}^\ast_{1,1} & 0  \\
                      -X^\ast_{4,1}& - X^\ast_{3,2} &   X^\ast_{2,1} & {{X}}^\ast_{1,2}\\
                      0 & -{ {X}}^\ast_{4,2} &  0 & { {X}}^\ast_{2,2}
                  \end{array}
 \right],
 \end{eqnarray}
where $[\begin{array}{cc}
       X_{1,1} &
        X_{1,2}
    \end{array}
   ]^t=\mathbf{\Theta} [\begin{array}{cc}
       S_{1} &
        S_{2}
    \end{array}
   ]^t$, $[\begin{array}{cc}
       X_{2,1} &
        X_{2,2}
    \end{array}
   ]^t=\mathbf{\Theta} [\begin{array}{cc}
       S_{3} &
        S_{4}
    \end{array}
   ]^t$, $
   [\begin{array}{cc}
       X_{3,1} &
        X_{3,2}
    \end{array}
   ]^t=\mathbf{\Theta} [\begin{array}{cc}
       S_{5} &
        S_{6}
    \end{array}
   ]^t$, and $
   [\begin{array}{cc}
       X_{4,1} &
        X_{4,2}
    \end{array}
   ]^t=\mathbf{\Theta} [\begin{array}{cc}
       S_{7} &
        S_{8}
    \end{array}
   ]^t$.  The constellation rotation matrix $\mathbf{\Theta}$ for this example can be chosen as
$$
\mathbf{\Theta}=\left[\begin{array}{rr}
                        \gamma  &  \delta  \\
                          -\delta  & \gamma
                      \end{array}
\right],
$$ where $\gamma=\cos \theta$ and $\delta=\sin \theta$ with $\theta=1.02$ \cite{Guo}.

In general,  the signal rotation matrix $\mathbf{\Theta}$ is designed   to achieve the signal space diversity. In this paper, we adopt the   optimal cyclotomic lattices design proposed in \cite{WangGY}. For $M$ transmit antennas, from \cite[Table I]{WangGY} we can get a set of integers $(m,n)$ and let $K=mn$. Then, the optimal lattice $\mathbf{\Theta}$ of size $M\times M$ is given by  \cite[Eq. (16)]{WangGY}
  \begin{eqnarray}\label{eqn:cyclo}
 \mathbf{\Theta}= \left[\begin{array}{llll}
                          \zeta_K & \zeta_K^2 & \cdots & \zeta_K^M \\
                         \zeta_{K}^{1+n_2m} & \zeta_{K}^{2(1+n_2m)} & \cdots & \zeta_{K}^{M(1+n_2m)}\\
                         \vdots & \vdots & \ddots & \vdots \\
                       \zeta_{K}^{1+n_Mm} & \zeta_{K}^{2(1+n_Mm)} & \cdots & \zeta_{K}^{M(1+n_Mm)}
                       \end{array}
 \right],
 \end{eqnarray}
 where
%$\zeta_n=\exp(\mathbf{j}2\pi/n)$,
$\zeta_K=\exp(\mathbf{j}2\pi/K)$
with $\mathbf{j}=\sqrt{-1}$ and $n_2,n_3, \cdots, n_M$ are distinct integers such that $1+n_im$ and $K$ are co-prime for any $2\leq i \leq M$.

\begin{remark}[Code Design for $M=2m-1$ ]
When $M$ is odd, the proposed code design can be obtained by extracting $M$ columns of the codeword of $\mathbf{\mathbf{B}}_{M+1,T,P}$. This is equivalent to transmitting nothing via the $(M+1)$-th antenna  using the code $\mathbf{\mathbf{B}}_{M+1,T,P}$.
\end{remark}

\begin{remark}[Code Rate]
For $M$ even,  $MP$ independent information symbols are sent over $T$ time slots and $T=2P+M-2$. Hence, the rate is
  \begin{eqnarray}
  R=\frac{MP}{2P+M-2}
 \end{eqnarray} symbols per channel use. For very a large $P$, the rate can be up to $M/2$. For a very large $M$,  the rate can be up to $P$.
\end{remark}

\begin{remark}[PIC Group Decoding Complexity]
The PIC group decoding complexity is related to the number of symbols to be jointly ML decoded in one group.
The ML decoding complexity in the PIC group decoding is $\sum^{P}_{p=1}{|\mathcal{A}|^{l_p}}$ \cite{ZhangGC09}, where $l_p$ is the number of symbols in the $p$-th group.
Our proposed code in (\ref{eqn:new}) reduces the decoding complexity due to its  group orthogonality similar to Alamouti code \cite{Alamouti}.  Therefore, the decoding complexity is reduced to $\sum^{2P}_{p=1}{|\mathcal{A}|^{l_p/2}}$.

\end{remark}

\begin{remark}[Comparison with linear receiver based STBC]
   It should be noted that if $X_{p,1}=X_{p,2}=\cdots,X_{p,\frac{M}{2}}$ for all $p=1,2,\cdots, 2P$, i.e., (\ref{eqn:wei}) is a Toeplitz matrix, then the proposed STBC in (\ref{eqn:new}) is very similar to the one in \cite{WangHM} (in \cite{WangHM}, the time
    reversal for the information symbols is used, while here it is
    not used). However, the rate of the linear receiver based STBC is not above $1$.

\end{remark}

\subsection{Achieving Full Diversity with ML Decoding}

Next, we shall show that the proposed STBC in (\ref{eqn:new}) can collect the full diversity with ML decoding.

\begin{theorem}\label{theorem1}
Consider a MIMO system with $M$ transmit antennas and $N$ receive antennas over block fading channels. The  STBC $\mathbf{\mathbf{B}}_{M,T,P}$ as described  in (\ref{eqn:new}) achieves full diversity under the ML decoding.
\end{theorem}

\begin{proof}[Proof of Theorem \ref{theorem1}]

In order to prove that ST code $\mathbf{\mathbf{B}}_{M,T,P}$ can obtain full diversity under ML decoding, it is sufficient to prove that $\breve{\mathbf{B}}=\mathbf{B}-\hat{\mathbf{B}}$ achieves full rank for any distinct pair of ST codewords $\mathbf{B}$ and $\hat{\mathbf{B}}$. For any pair of distinct codewords $\mathbf{B}$ and $\hat{\mathbf{B}}$, there exits at least one  vector $\mathbf{X}^i_p$ such that $\mathbf{X}^i_p-\hat{\mathbf{X}}^i_p\neq \mathbf{0}$, where $\mathbf{X}^i_p$ and $\hat{\mathbf{X}}^i_p$ are related to $\mathbf{s}^i_p$ and $\hat{\mathbf{s}}^i_p$ from (\ref{eqn:RM}), respectively. Define $\check{X}=X-\hat{X}$ as the difference between symbols $X$ and $\hat{X}$. Then, we can further deduce that no any element in the vector $\check{\mathbf{X}}^i_p$ can be zero because the signal space diversity is obtained from the signal rotation in (\ref{eqn:RM}) \cite{WangGY}.

Observing the proposed code (\ref{eqn:Tx4a}) for $M=4$,  we can get the codeword different matrix as follows:
\begin{eqnarray}\label{eqn:Tx4b}
\breve{\mathbf{B}}_{4,6,2} = \left[\begin{array}{cccc}
                     \check{X}_{1,1} & 0  & \check{X}_{3,1}  &  0 \\
                     \check{X}_{2,1} & \check{X}_{1,2} & \check{X}_{4,1}& \check{X}_{3,2} \\
                      0 & \check{X}_{2,2} & 0 & \check{X}_{4,2} \\
                     -{\check{X}}^\ast_{3,1} & 0 & {\check{X}}^\ast_{1,1} & 0  \\
                      -{\check{X}}^\ast_{4,1}& -{\check{X}}^\ast_{3,2} &  {\check{X}}^\ast_{2,1} & {\check{X}}^\ast_{1,2}\\
                      0 & -{\check{X}}^\ast_{4,2} &  0 & {\check{X}}^\ast_{2,2}
                  \end{array}
 \right].
 \end{eqnarray}
After  permutating  rows and columns of $\breve{\mathbf{B}}_{4,6,2}$, we have
\begin{eqnarray}
\breve{\mathbf{B}}_{4,6,2}^{'} = \left[\begin{array}{cccc}\label{tran}
                     \check{X}_{1,1} & \check{X}_{3,1}   & 0 &  0 \\
                      -{\check{X}}^\ast_{3,1} & {\check{X}}^\ast_{1,1} & 0 & 0 \\
                     \check{X}_{2,1} & \check{X}_{4,1} & \check{X}_{1,2} & \check{X}_{3,2} \\
                      -{\check{X}}^\ast_{4,1} & {\check{X}}^\ast_{2,1} & -{\check{X}}^\ast_{3,2}  & {\check{X}}^\ast_{1,2}  \\
                      0& 0 &  \check{X}_{2,2} & \check{X}_{4,2}\\
                      0 & 0 &  -{\check{X}}^\ast_{4,2} & {\check{X}}^\ast_{2,2}
                  \end{array}
 \right].
 \end{eqnarray}
Note that $\breve{\mathbf{B}}_{4,6,2}$ and $\breve{\mathbf{B}}_{4,6,2}^{'}$ have the same rank, because rows/columns permutations do not change the rank of the matrix.

Similarly, we can write a codeword difference matrix $\breve{\mathbf{B}}_{M,T,P}^{'}$ of the proposed code in (\ref{eqn:new}) as follows: (after some row/column permutations)
\begin{eqnarray}\label{eqn:MatPhi}
 \breve{\mathbf{B}}_{M,T,P}^{'}
 = \left[\begin{array}{cccc}
                     \check{\mathbf{T}}_{1,1} & \mathbf{0}  & \cdots  &  \mathbf{0} \\
                     \check{\mathbf{T}}_{2,1} & \check{\mathbf{T}}_{1,2} & \ddots  & \vdots  \\
                      \vdots & \check{\mathbf{T}}_{2,2} &  \ddots &  \mathbf{0} \\
                     \vdots & \vdots & \ddots & \check{\mathbf{T}}_{1,\frac{M}{2}}  \\
                      \vdots &  \vdots  & \ddots  & \check{\mathbf{T}}_{2,\frac{M}{2}} \\
                     \check{\mathbf{T}}_{P,1} & \vdots   & \ddots  & \vdots  \\
                     \mathbf{0} &  \check{\mathbf{T}}_{P,2} &  \ddots & \vdots\\
                     \vdots & \mathbf{0} & \ddots & \vdots\\
                     \mathbf{0} & \vdots & \ddots & \check{\mathbf{T}}_{P,\frac{M}{2}}
                  \end{array}
                 \right] ,
 \end{eqnarray}
where $\mathbf{0}$ is in fact $\mathbf{0}_{2\times 2}$ and
\begin{eqnarray}
\ {\check{\mathbf{T}}_{i,j}} = \left[\begin{array}{cc}
    \check{X}_{i,j} & \check{X}_{P+i,j} \\
    -\check{X}^\ast_{P+i,j} & \check{X}^\ast_{i,j}\\
  \end{array}
  \right]
 \end{eqnarray} for $i=1,\cdots, P$ and $j=1,\cdots,\frac{M}{2}$.

Because all elements in  $\check{\mathbf{X}}^i_p$ are nonzero,     matrices $\{\check{\mathbf{T}}_{p,1},\check{\mathbf{T}}_{p,2},\cdots,\check{\mathbf{T}}_{p,\frac{M}{2}}\}$ must be all nonzero. Assume  that $p$ is the minimal  index such that  $\check{\mathbf{X}}^i_p\neq \mathbf{0}$. Then, we can write (\ref{eqn:MatPhi}) as
\begin{eqnarray}
 \breve{\mathbf{B}}_{M,T,P}^{'}
 = \left[\begin{array}{llll}
                   \mathbf{0}  &\mathbf{0} &\cdots&\mathbf{0} \\
                   \vdots&\vdots&\ddots&\vdots\\
                   \mathbf{0}  &\mathbf{0} &\cdots&\mathbf{0} \\
                   --& -- & -- & --\\
                     \check{\mathbf{T}}_{p,1} & \mathbf{0}  & \cdots  &  \mathbf{0} \\
                     \check{\mathbf{T}}_{p+1,1} & \check{\mathbf{T}}_{p,2} & \ddots  & \vdots  \\
                      \vdots & \check{\mathbf{T}}_{p+1,2} &  \ddots &  \mathbf{0} \\
                     \vdots & \vdots & \ddots & \check{\mathbf{T}}_{p,\frac{M}{2}}  \\
                     --& -- & -- & --\\
                      \vdots &  \vdots  & \ddots  & \check{\mathbf{T}}_{p+1,\frac{M}{2}} \\
                     \check{\mathbf{T}}_{P,1} & \vdots   & \ddots  & \vdots  \\
                     \mathbf{0} &  \check{\mathbf{T}}_{P,2} &  \ddots & \vdots\\
                     \vdots & 0 & \ddots & \vdots\\
                     \mathbf{0} & \vdots & \ddots & \check{\mathbf{T}}_{P,\frac{M}{2}}
                  \end{array}
                 \right]
                  \triangleq \left[\begin{array}{c}
                  \mathbf{0}\\
                  --\\
                    \check{\mathbf{D}}_1 \\
                       --\\
                   \check{\mathbf{D}}_2 \\
                  \end{array}
                  \right].\nonumber
 \end{eqnarray}
We further have
\begin{eqnarray}
 &&\det\left((\breve{\mathbf{B}}_{M,T,P}^{'})^H \breve{\mathbf{B}}_{M,T,P}^{'}\right)\nonumber\\ %\\&=& \det(\check{\mathbf{D}}^H_1\check{\mathbf{D}}_1)+\det(\check{\mathbf{D}}^H_2\check{\mathbf{D}}_2)\\
  &\geq&  \det(\check{\mathbf{D}}^H_1)\det(\check{\mathbf{D}}_1)\\
  &\overset{(a)}{=}& \prod_{j=1,\cdots,M/2} \left(|\check{X}_{p,j}|^2+|\check{X}_{P+p,j}|^2\right)^2 \\&>&0\nonumber,
 \end{eqnarray}
where  $\overset{(a)}{=}$ is obtained from  the property of block diagonal matrix $\det(\check{\mathbf{D}}_1)=\prod_{j=1,2,\cdots,\frac{M}{2}}\det(\check{\mathbf{T}}_{p,j})$.

 Therefore, for any nonzero $\mathbf{X}^i_p-\hat{\mathbf{X}}^i_p$  our proposed codes in (\ref{eqn:new}) can achieve full diversity with ML decoding.
\end{proof}

\subsection{Achieving Full diversity with PIC Group Decoding when $P=2$}

Next, we show that the proposed STBC can obtain full diversity when a PIC group decoding with a particular grouping scheme is used at the receiver.

\begin{theorem}\label{theorem2}
Consider a MIMO system with $M$ transmit antennas and $N$ receive antennas over block fading channels. The  STBC  $\mathbf{\mathbf{B}}_{M,T,P}$ as described in (\ref{eqn:new}) with two diagonal layers (i.e., $P=2$)  is used at the transmitter. The equivalent channel matrix is $\mathcal{H}\in \mathcal{C}^{TN\times MP}$. If the received signal is decoded using the PIC group decoding with the grouping scheme $\mathcal{I}={\{\mathcal{I}_1,\mathcal{I}_2,\mathcal{I}_3,\mathcal{I}_4\}}$, where $\mathcal{I}_p=\{(p-1)M/2+1,\cdots,pM/2\}$ for $p=1,2,3,4$, i.e., the size of each group is equal to  $M/2$, then the code $\mathbf{\mathbf{B}}_{M,T,P}$ achieves the full diversity.   The code rate of the full-diversity STBC can be up to $2$ symbols per channel use.
%The code of the full-diversity STBC $\mathbf{\mathbf{B}}_{M,T,P}$ has the same rate as the elementary code $\mathbf{C}^{i}_{M,T,P}$, and can be up to %2 symbols per channel use.
\end{theorem}

In order to prove Theorem \ref{theorem2}, let us first introduce the following lemma.

%%%\begin{lemma}   Consider a matrix group $[A, B, C, D]$, $A, B, C, D$ are m-by-n matrices. If $A$ and $B$ are linearly independent, and $A$ is orthogonal to C, and $B$ is orthogonal to $D$, we always have that $A$, $B$, $C$ and $D$ are linearly independent vector groups.
%%%\end{lemma}

\begin{lemma}
Consider the system as described in \emph{Theorem 2} with $N=1$ and STBC $\mathbf{\mathbf{B}}_{M,T,P}$ as given by (\ref{eqn:new}).
 %\begin{description}
 %                   \item[1)]
                   The equivalent channel matrix $\mathcal{H}\in \mathcal{C}^{TN\times MP}$ of the code $\mathbf{\mathbf{B}}_{M,T,P}$ can be expressed as
  \begin{eqnarray}\label{eqn:equiH}
% \mathcal{H}&=&
&& \left[\begin{array}{cccccccc}
                     \mathcal{G}^1_1 & \mathcal{G}^1_2 & \cdots & \mathcal{G}^1_P & \mathcal{G}^2_1 & \mathcal{G}^2_2 & \ldots & \mathcal{G}^2_P \nonumber\\
                      -(\mathcal{G}^2_1)^\ast & -(\mathcal{G}^2_2)^\ast & \cdots & -(\mathcal{G}^2_P)^\ast & (\mathcal{G}^1_1)^\ast & (\mathcal{G}^1_2)^\ast & \cdots & (\mathcal{G}^1_P)^\ast
                  \end{array}
                 \right] \nonumber\\
&&\triangleq \left[\begin{array}{cccc}
                     \mathbf{G}_1 & \mathbf{G}_2 &
                     \cdots & \mathbf{G}_{2P}
                  \end{array}
                 \right],
 \end{eqnarray}
 where $\mathcal{G}^i_p$ is given by
  \begin{eqnarray}\label{eqn:Gip}
 \mathcal{G}^i_p = \left[\begin{array}{c}
                     \mathbf{0}_{(p-1)\times M/2} \\
                     \mathrm{diag}(\mathbf{h}_i)\mathbf{\Theta} \\
                      \mathbf{0}_{(P-p)\times M/2}
                  \end{array}
                 \right], \,\,\,\, i=1,2;\,\,\,p=1,2,\cdots,P
 \end{eqnarray}
 with $\mathbf{h}_1=[\begin{array}{cccc}
                       h_1 & h_2 & \cdots & h_{\frac{M}{2}}
                     \end{array}
 ]^t$ and $\mathbf{h}_2=[\begin{array}{cccc}
                       h_{\frac{M}{2}+1} & h_{\frac{M}{2}+2} & \cdots & h_{M}
                     \end{array}
 ]^t$.  $h_j$ denotes the channel between the $j$th transmit antenna and the single receive antenna for $j=1,2,\cdots, M$.

% \item[2)] For $P=2$, $\mathcal{H}=[\begin{array}{cccc}
 %                                                      \mathbf{G}_1 & \mathbf{G}_2 & \mathbf{G}_3 & \mathbf{G}_4
 %                                                    \end{array}
 %                  ]$ where  vector groups $\mathbf{G}_1, \mathbf{G}_2, \mathbf{G}_3$ and $\mathbf{G}_4$ are linearly independent.
 %                \end{description}

 \end{lemma}

The proof of \emph{Lemma 1} is in Appendix.

Now we are ready to prove \emph{Theorem \ref{theorem2}}.
\begin{proof}[Proof of  Theorem \ref{theorem2}]
Let us consider $N=1$ first.
For a MISO system with $P=2$, from \emph{Lemma 1} the equivalent channel matrix of the proposed code $\mathbf{\mathbf{B}}_{M,T,P}$  is given by
\begin{eqnarray}\label{eqn:Hprove}
 \mathcal{H}&=& \left[\begin{array}{cccc}
                     \mathcal{G}^1_1 & \mathcal{G}^1_2 &   \mathcal{G}^2_1 & \mathcal{G}^2_2   \\
                      -(\mathcal{G}^2_1)^\ast & -(\mathcal{G}^2_2)^\ast & (\mathcal{G}^1_1)^\ast & (\mathcal{G}^1_2)^\ast
                  \end{array}
                 \right]\\
                 &\triangleq &\left[\begin{array}{cccc}
                     \mathbf{G}_1 & \mathbf{G}_2 &
                     \mathbf{G}_{3} & \mathbf{G}_{4}
                  \end{array}
                 \right],
 \end{eqnarray}
where $\mathcal{G}^i_p$ is given by (\ref{eqn:Gip}), for $i=1,2,\,\, p=1,2$.

Denote $\mathbf{f}_i$ a length-$\frac{M}{2}$ row vector, given by $\mathbf{f}_i=h_i \mathbf{\Upsilon}_i$ for $i=1,2,\cdots, M$, where $\mathbf{\Upsilon}_i$ is the $i$-th row of the following matrix:
\begin{eqnarray}\label{eqn:Upsilon}
\mathbf{\Upsilon} = \left[\begin{array}{c}
                            \mathbf{\Theta} \\
                             \mathbf{\Theta}
                          \end{array}
\right]
\end{eqnarray}
with $\mathbf{\Theta}$ being the rotation matrix of size $\frac{M}{2}\times \frac{M}{2}$.

%For $P=2$, from (\ref{eqn:equiH}) the equivalent channel matrix of the code in (\ref{eqn:new}) can be written as
Then, (\ref{eqn:Hprove}) can be written as
\begin{eqnarray}
\mathcal{H} &=&
\left[\begin{array}{cccc}
                 \mathbf{G}_1 & \mathbf{G}_2  & \mathbf{G}_3  & \mathbf{G}_4
               \end{array}
               \right]
\\\label{eqn:HH}
&=& \left[\begin{array}{cccc}
                      \mathbf{f}_1 & \mathbf{0} & \mathbf{f}_{\frac{M}{2}+1}  & \mathbf{0} \\
                      \mathbf{f}_2 & \mathbf{f}_1 & \mathbf{f}_{\frac{M}{2}+2}  & \mathbf{f}_{\frac{M}{2}+1} \\
                           \vdots  &  \mathbf{f}_2&  \vdots &  \mathbf{f}_{\frac{M}{2}+2} \\
                      \mathbf{f}_{\frac{M}{2}} & \vdots &\mathbf{f}_M& \vdots \\
                     \mathbf{0}  &   \mathbf{f}_{\frac{M}{2}}    &  \mathbf{0} &  \mathbf{f}_{M} \\
                        -\mathbf{f}_{\frac{M}{2}+1}^*  & \mathbf{0}   &   \mathbf{f}_1^* & \mathbf{0}  \\
                        -\mathbf{f}_{\frac{M}{2}+2}^*  & -\mathbf{f}_{\frac{M}{2}+1}^*    &   \mathbf{f}_2^* & \mathbf{f}_1^*   \\
                        \vdots &  -\mathbf{f}_{\frac{M}{2}+2}^*   &    \vdots  &  \mathbf{f}_2^*   \\
                        -\mathbf{f}_M^* & \vdots   &  \mathbf{f}_{\frac{M}{2}}^* & \vdots  \\
                        \mathbf{0} &  -\mathbf{f}_{M}^*    &    \mathbf{0}  &   \mathbf{f}_{\frac{M}{2}}^*
                    \end{array}
\right],
\end{eqnarray}
where $\mathbf{0}=\mathbf{0}_{1\times\frac{M}{2}}$.

After some row/column permutations of (\ref{eqn:HH}), we can get
\begin{eqnarray}\label{eqn:HHH}
\mathcal{H}^{'} &=& \left[\begin{array}{cccc}
                     [\mathbf{G}_1^{'}  & \mathbf{G}_3^{'} ]& [\mathbf{G}_2^{'} & \mathbf{G}_4^{'}]
                   \end{array}
\right]\nonumber\\
&=&\left[\begin{array}{cc}
                       \mathcal{F}_1 & \mathbf{0}_{2\times M} \\
                        \mathcal{F}_2 & \mathcal{F}_1  \\
                        \vdots &  \mathcal{F}_2  \\
                        \mathcal{F}_{\frac{M}{2}} &\vdots  \\
                        \mathbf{0}_{2\times M} & \mathcal{F}_{\frac{M}{2}}
                      \end{array}
\right],
\end{eqnarray}
where  $\mathcal{F}_j$ is a $2\times M$ matrix given by
\begin{eqnarray}
\mathcal{F}_j=\left[\begin{array}{cc}
                      \mathbf{f}_j & \mathbf{f}_{j+\frac{M}{2}} \\
                     -\mathbf{f}_{j+\frac{M}{2}}^* &   \mathbf{f}_j^*
                   \end{array}
\right]
\end{eqnarray}
for $j=1,2,\cdots, \frac{M}{2}$.

Next, we shall  prove that any non-zero linear combination of the vectors in $\mathbf{G}_1^{'}$  over $\Delta \mathcal{A}$ does not belong to the space linearly spanned by all the vectors in the vector groups $[\mathbf{G}_2^{'}\,\,\mathbf{G}_4^{'}]$ as long as $\mathbf{h}\neq 0$, i.e., for $a_i\in \Delta \mathcal{A}$  not all zero, and $c_j\in \mathbb{C}$
\begin{eqnarray}\label{eqn:true}
\sum_{\forall  \mathbf{g}_{i}\subset \mathbf{G}_1^{'}} a_i \mathbf{g}_{i} \neq \sum_{\forall \mathbf{g}_{j}\subset \{\mathbf{G}_2^{'}, \mathbf{G}_4^{'}\}} c_j \mathbf{g}_{j},
\end{eqnarray}
where $\mathbf{g}_{i}$ is a column vector.

 To prove (\ref{eqn:true}), we use the self-contradiction method as follows.

 Suppose that    for $a_i\in \Delta \mathcal{A}$  not all zero, and $c_j\in \mathbb{C}$
 \begin{eqnarray}\label{eqn:assump}
\sum_{\forall  \mathbf{g}_{i}\subset \mathbf{G}_1^{'}} a_i \mathbf{g}_{i} = \sum_{\forall \mathbf{g}_{j}\subset \{\mathbf{G}_2^{'}, \mathbf{G}_4^{'}\}} c_j \mathbf{g}_{j},
\end{eqnarray}

For any $\mathbf{h}\neq \mathbf{0}$ with $\mathbf{h}=[\begin{array}{ccc}
                                      h_1 & \cdots & h_{M}
                                    \end{array}
]^t$, there exists the minimum index $q$ ($1\leq q\leq M$) such that $h_q\neq 0$ and $h_{v}=0, \forall v<q$. Then, we can find that the block $\mathcal{F}_j$ associated with $h_q$ is nonzero and that  blocks $\mathcal{F}_1,\cdots, \mathcal{F}_{j-1}$ must be all zeros. Therefore, (\ref{eqn:HHH}) can be expressed as
\begin{eqnarray}\label{eqn:HHH2}
\mathcal{H}^{'}
&=&\left[\begin{array}{cc}
                       \mathbf{0} & \mathbf{0} \\
                       \vdots & \vdots \\
                       \mathbf{0}& \mathbf{0} \\
                        \mathcal{F}_j & \mathbf{0} \\
                        \vdots &  \mathcal{F}_j  \\
                        \mathcal{F}_{\frac{M}{2}} &\vdots  \\
                        \mathbf{0} & \mathcal{F}_{\frac{M}{2}}
                      \end{array}
\right],
\end{eqnarray}
where $\mathbf{0}=\mathbf{0}_{2\times M}$.

Using (\ref{eqn:assump}) and examining the $(2j-1)$th row of (\ref{eqn:HHH2}), we get
\begin{eqnarray}\label{eqn:cond1}
\mathbf{f}_j \cdot  \mathbf{a}  =0,
\end{eqnarray}
where $\mathbf{a}  =[\begin{array}{ccc}
                                                           a_1 & \cdots & a_{M/2}
                                                         \end{array}
]^t$ and $a_k\in \Delta \mathcal{A}$  not all zero, $k=1,\cdots, M/2$.

Recall $\mathbf{f}_j=h_j\mathbf{\Upsilon}_j$, where $\mathbf{\Upsilon}_j$ is the $j$th row of the matrix $\mathbf{\Upsilon}$ in (\ref{eqn:Upsilon}). We can rewrite (\ref{eqn:cond1}) as
\begin{eqnarray}\label{eqn:cond2}
h_j \sum_{k=1}^{M/2} a_{k} \theta_{j,k}=0,
\end{eqnarray}
for  $a_k\in \Delta \mathcal{A}$  not all zero, where $\theta_{j,k}$ is the ($j,k$)-th entry of the matrix $\mathbf{\Upsilon}$.

Note that the rotation matrix $\mathbf{\Theta}$ in (\ref{eqn:cyclo}) is designed so that $\sum_{k=1}^{M/2} a_{k} \theta_{j,k}\neq0$, for $a_k\in \Delta \mathcal{A}$  not all zero. It contradicts the result (\ref{eqn:cond2}) based on the assumption of (\ref{eqn:assump}).
Hence, (\ref{eqn:true}) holds, i.e., any non-zero linear combination of the vectors in $\mathbf{G}_1^{'}$  over $\Delta \mathcal{A}$ does not belong to the space linearly spanned by all the vectors in the vector groups $[\mathbf{G}_2^{'}\,\,\mathbf{G}_4^{'}]$. Furthermore, we can see that  vector group $\mathbf{G}_1^{'}$ is orthogonal to $\mathbf{G}_3^{'}$. We then conclude that any non-zero linear combination of the vectors in $\mathbf{G}_1^{'}$  over $\Delta \mathcal{A}$ does not belong to the space linearly spanned by all the vectors in the vector groups $[\mathbf{G}_2^{'}\,\,\mathbf{G}_3^{'}\,\,\mathbf{G}_4^{'}]$.

Similarly, we can prove that any non-zero linear combination of the vectors in $\mathbf{G}_p^{'}$  over $\Delta \mathcal{A}$ does not belong to the space linearly spanned by all the vectors in the remaining vector groups, for $p=2, 3, 4$.

Note that $\mathbf{G}_m^{'}$ is a row permutation of $\mathbf{G}_m$ for $m=1,2,3,4$, respectively. We   prove that any non-zero linear combination of the vectors in $\mathbf{G}_m$  over $\Delta \mathcal{A}$ does not belong to the space linearly spanned by all the vectors in the remaining vector groups, for $m=1, 2, 3, 4$.

%%%We can see that vector groups $\left[\mathbf{G}_1^{'}, \mathbf{G}_3^{'} \right]$ are linearly independent of  $\left[\mathbf{G}_2^{'}, \mathbf{G}_4^{'} \right]$.   The proof is similar to the proof of \emph{Lemma 1.2}  in \cite{ZhangXuXia}.
%%%
%%%Moreover, vector group $\mathbf{G}_1^{'}$ is orthogonal to $\mathbf{G}_3^{'}$. Thus, $\mathbf{G}_1^{'}$ cannot be expressed by any linear combination of rest column vectors in $\mathcal{H}^{'}$.
%%%Similarly, we can find that $\mathbf{G}_2^{'}$, $\mathbf{G}_3^{'}$ and $\mathbf{G}_4^{'}$ cannot be written as any linear combination of rest column vectors, respectively. Therefore, vector groups $\mathbf{G}_1^{'}$, $\mathbf{G}_2^{'}$, $\mathbf{G}_3^{'}$ and $\mathbf{G}_4^{'}$ are linearly independent.  Note $\mathbf{G}_m^{'}$ is a row permutation of $\mathbf{G}_m$ for $m=1,2,3,4$. Hence, vector groups $\mathbf{G}_1$, $\mathbf{G}_2$, $\mathbf{G}_3$ and $\mathbf{G}_4$ are linearly independent.

In the above, we prove that for the STBC (\ref{eqn:new}) with PIC group decoding the second condition in \emph{Proposition \ref{prop1}}  is satisfied when there is only one receive antenna.
For $N>1$, the equivalent channel matrix will be a stacked matrix of (\ref{eqn:HH}) with the  number of columns unchanged. It is easy to see that when there are multiple receive antennas, the second condition of \emph{Proposition \ref{prop1}} is also satisfied. The proof of \emph{Theorem \ref{theorem2}} is completed.
\end{proof}

\subsection{Achieving Full Diversity with PIC-SIC Group Decoding for
Any $P$}

For the proposed STBC (\ref{eqn:new}) with any number of layers and the PIC-SIC group decoding we have the following results.

 \begin{theorem}\label{theorem3}
 Consider a MIMO system with $M$ transmit antennas and $N$ receive antennas over block fading channels. The  STBC  as described in (\ref{eqn:new}) with $P$ diagonal layers is used at the transmitter. The equivalent channel matrix is $\mathcal{H}\in \mathcal{C}^{TN\times MP}$. If the received signal is decoded using the PIC-SIC group decoding with the grouping scheme $\mathcal{I}={\{\mathcal{I}_1, \mathcal{I}_2, \cdots,  \mathcal{I}_{2P}\}}$ and with the \emph{sequential order} $\{1,2,\cdots, 2P\}$, where $\mathcal{I}_p=\{(p-1)M/2+1,\ldots,pM/2\}$ for $p=1,2,\cdots,2P$, i.e., the size of each group is equal to  $M/2$, then the code $\mathbf{\mathbf{B}}_{M,T,P}$ achieves the full diversity. The code rate of the full-diversity STBC can be up to $M/2$ symbols per channel use.
 \end{theorem}

The proof is similar to that of \emph{Theorem \ref{theorem2}}.  Note that  $\mathcal{H}$ for the code $\mathbf{B}_{M,T,P}$ in \emph{Lemma 1} can be written as an alternative form similar to the one in (\ref{eqn:HHH}) except the expansion of column dimensions. It is then not hard to follow the  proof for the case of $P=2$ in Section III-C to prove \emph{Theorem \ref{theorem3}} by showing that the criterion in
{\em Proposition \ref{prop2}} is satisfied.
The detailed proof is omitted.

\section{CODE DESIGN EXAMPLES}\label{sec:new}
In this section, we show a few examples of the proposed STBC given in (\ref{eqn:new}).

\subsection{For Four Transmit Antennas}

For $M=4$ and $T=6$. According to the code design in (\ref{eqn:new}), we have
\begin{equation}\label{eqn:F4}
\mathbf{\mathbf{B}}_{4,6,2} = \left[\begin{array}{cc}
    \mathbf{C}^1_{2,3,2} & \mathbf{C}^2_{2,3,2} \\
    -(\mathbf{C}^2_{2,3,2})^\ast & (\mathbf{C}^1_{2,3,2})^\ast\\
  \end{array}
  \right],
  \end{equation}
  where
\begin{equation*}
\mathbf{C}^1_{2,3,2} = \left[\begin{array}{cc}
    X_{1,1} & 0 \\
    X_{2,1} & X_{1,2}\\
    0 & X_{2,2}
  \end{array}
  \right],
  {\rm{~and~}}
  \mathbf{C}^2_{2,3,2} = \left[\begin{array}{cc}
    X_{3,1} & 0 \\
    X_{4,1} & X_{3,2}\\
    0 & X_{4,2}
  \end{array}
  \right].
 \end{equation*}

The code rate of $\mathbf{\mathbf{B}}_{4,6,2}$ is $4/3$.

The equivalent channel matrix of the code $\mathbf{\mathbf{B}}_{4,6,2}$ is
\begin{equation}\label{eqn:H4}
%\mathcal{H}=
 \left[\begin{array}{cccccccc}
\gamma h_1 & \delta h_1 & 0 & 0 & \gamma h_3 & \delta h_3 & 0 & 0\\
-\delta h_2 & \gamma h_2 & \gamma h_1 & \delta h_1 & -\delta h_4 & \gamma h_4 & \gamma h_3 & \delta h_3 \\
0 & 0 & -\delta h_2 & \gamma h_2 & 0 & 0 & -\delta h_4 & \gamma h_4\\
-\gamma h_3^\ast & -\delta h_3^\ast & 0 & 0 & \gamma h_1^\ast & \delta h_1^\ast  & 0 & 0\\
\delta h_4^\ast & -\gamma h_4^\ast & -\gamma h_3^\ast & -\delta h_3^\ast & -\delta h_2^\ast & \gamma h_2^\ast & \gamma h_1^\ast & \delta h_1^\ast\\
0& 0& \delta h_4^\ast & -\gamma h_4^\ast& 0 & 0& -\delta h_2^\ast & \gamma h_2^\ast
\end{array}
 \right],
  \end{equation}
where  $\gamma=\cos \theta$ and $\delta=\sin \theta$ with $\theta=1.02$ \cite{Guo}.

To achieve the full diversity, the grouping scheme for the PIC group decoding of the code $\mathbf{\mathbf{B}}_{4,6,2}$ is $\mathcal{I}_1=\{1,2\}$, $\mathcal{I}_2=\{3,4\}$, $\mathcal{I}_3=\{5,6\}$, and $\mathcal{I}_4=\{7,8\}$.

\subsection{For Eight Transmit Antennas}

For given $T=10$, the code achieving full diversity with PIC group decoding can be designed as follows,
 \begin{equation}\label{eqn:F8}
\mathbf{\mathbf{B}}_{8,10,2} = \left[\begin{array}{cc}
    \mathbf{C}^1_{4,5,2} & \mathbf{C}^2_{4,5,2} \\
    -(\mathbf{C}^2_{4,5,2})^\ast & (\mathbf{C}^1_{4,5,2})^\ast\\
  \end{array}
  \right],
  \end{equation}
  where
\begin{eqnarray}
\mathbf{C}^1_{4,5,2} &=& \left[\begin{array}{cccc}
                     X_{1,1} &    &    &   \\
                    X_{2,1}& X_{1,2} &    &   \\
                     & X_{2,2} &  X_{1,3} &   \\
                      &  & X_{2,3} & X_{1,4}  \\
                       &   &   & X_{2,4} \\
                  \end{array}
 \right]\\
\mathbf{C}^2_{4,5,2} &=& \left[\begin{array}{cccc}
                     X_{3,1} &    &    &   \\
                    X_{4,1}& X_{3,2} &    &   \\
                     & X_{4,2} &  X_{3,3} &   \\
                      &  & X_{4,3} & X_{3,4}  \\
                       &   &   & X_{4,4} \\
                  \end{array}
 \right].
 \end{eqnarray}

The code rate of $\mathbf{\mathbf{B}}_{8,10,2}$ is $8/5$.

The equivalent channel of the code $\mathbf{\mathbf{B}}_{8,10,2}$ is
\begin{eqnarray}\label{eqn:H8}
\mathcal{H}=\left[\begin{array}{cccc}
              \mathbf{f}_1 & \mathbf{0} & \mathbf{f}_5 & \mathbf{0} \\
              \mathbf{f}_2 & \mathbf{f}_1 & \mathbf{f}_6 & \mathbf{f}_5 \\
              \mathbf{f}_3 & \mathbf{f}_2 & \mathbf{f}_7 & \mathbf{f}_6 \\
              \mathbf{f}_4 & \mathbf{f}_3 & \mathbf{f}_8 & \mathbf{f}_7 \\
              \mathbf{0} & \mathbf{f}_4 & \mathbf{0} & \mathbf{f}_8 \\
               -\mathbf{f}_5^* & \mathbf{0} & \mathbf{f}_1^* & \mathbf{0}\\
              -\mathbf{f}_6^* & -\mathbf{f}_5^* & \mathbf{f}_2^* & \mathbf{f}_1^* \\
              -\mathbf{f}_7^* & -\mathbf{f}_6^*  & \mathbf{f}_3^* & \mathbf{f}_2^* \\
              -\mathbf{f}_8^* & -\mathbf{f}_7^*  & \mathbf{f}_4^* & \mathbf{f}_3^*  \\
              \mathbf{0} & -\mathbf{f}_8^* &  \mathbf{0} & \mathbf{f}_4^*
            \end{array}\right],
\end{eqnarray}
where $\mathbf{0}=\mathbf{0}_{1\times 4}$ and $\mathbf{f}_i=h_i\mathbf{\Upsilon}_i$
 with $\mathbf{\Upsilon}_i$ being the $i$th row of the matrix $\mathbf{\Upsilon}=\left[ \begin{array}{cc}
                                                                                          1 & 1
                                                                                        \end{array}
 \right]^t\otimes\mathbf{ \Theta}$ for $i=1,2,\cdots, 8$ and $\mathbf{\Theta}$ being a rotation matrix of $4$-by-$4$.

 To achieve the full diversity, the grouping scheme for the PIC group decoding of the code $\mathbf{\mathbf{B}}_{8,10,2}$  is $\mathcal{I}_1=\{1,2,3,4\}$, $\mathcal{I}_2=\{5,6,7,8\}$, $\mathcal{I}_3=\{9,10,11,12\}$, and $\mathcal{I}_4=\{13,14,15,16\}$.

Table I  shows the decoding complexity of the new codes compared with the codes in \cite{ZhangGC09}. It can be seen that the proposed STBC in (\ref{eqn:new}) have more groups than the STBC in \cite{ZhangGC09}, but each group has half number of symbols to be jointly coded. Therefore, the PIC group decoding complexity is reduced by half. This mainly attributes to the introduction of Alamouti code structure into the code design in (\ref{eqn:new}) and the group orthogonality in the code matrix can  reduce the decoding complexity without sacrificing any performance benefits.

\begin{table}[t]
\begin{center}\caption{Comparison in PIC Group Decoding Complexity}
\renewcommand{\arraystretch}{1.5}
%\caption{Values of Parameters For Various Modulation Schemes}
\label{1} \centering
\begin{tabular}{c c c c c}
\hline \hline
\bfseries Code   & \bfseries  Groups & \bfseries Symbols/Group & \bfseries Decoding Complexity\\
\hline
$\mathbf{C}_{4,5,2}/\mathbf{C}_{4,6,2}$ \cite{ZhangGC09}   & $2$ & $4$ & $2{|\mathcal{A}|^4}$\\
$\mathbf{\mathbf{B}}_{4,6,2}$   & $4$ & $2$ & $4{|\mathcal{A}|^2}$\\
\hline
$\mathbf{C}_{8,9,2}/\mathbf{C}_{8,10,2}$ \cite{ZhangGC09}   & $2$ & $8$ & $2{|\mathcal{A}|^8}$\\
$\mathbf{\mathbf{B}}_{8,10,2}$   & $4$ & $4$ & $4{|\mathcal{A}|^4}$\\
\hline\hline
\end{tabular}
\end{center}
\end{table}

 \section{Simulation Results}
In this section, we present some simulation results of the proposed STBC and compare it to the existing codes.  Flat MIMO Rayleigh fading channels are considered.

Fig. 1 shows the bit error rate (BER)  simulation results   of the proposed code $\mathbf{\mathbf{B}}_{4,6,2}$ in (\ref{eqn:F4}) with ZF, BLAST \cite{Foschini} detection, PIC group decoding and ML detection in a $4\times 4$ system, respectively.   The signal modulation is 16QAM and then the bandwidth efficiency is $8$ bps/Hz. It is clearly that the ML decoding achieves the best BER performance. The PIC group decoding can obtain the full diversity as the ML decoding does, but it suffers a loss of coding gain around 1 dB. Neither ZF nor BLAST detection method  can  obtain the full diversity.

Fig. 2 illustrates the performance comparison among the proposed code $\mathbf{\mathbf{B}}_{4,6,2}$, the TAST code \cite{gamal}, and the perfect ST code \cite{oggier} for a $4\times 4$ system. It is known that both TAST and perfect ST codes can obtain full diversity and full rate (rate-$M$) for MIMO systems, but they are designed based on the ML decoding whose complexity is high, i.e., a joint $M^2$-symbol ML decoding.   When PIC group decoding is applied, it is seen from Fig. 2 that both TAST code and perfect ST code lose the diversity gain. For the proposed code $\mathbf{\mathbf{B}}_{4,6,2}$, the full diversity can be obtained for both ML and PIC group decoding. It should be mentioned that the code $\mathbf{\mathbf{B}}_{4,6,2}$ with PIC group decoding achieves the full diversity with a very low decoding complexity, i.e., a double-symbol ML decoding, much lower than that of TAST and perfect ST codes (a joint 16-symbol ML decoding).

Moreover, the proposed code  is compared with other code designs based on PIC group decoding such as codes $\mathbf{C}_{4,6,2}$ and $\mathbf{C}_{4,5,2}$  in \cite{ZhangGC09}, and Guo-Xia's code in \cite[Eq. (40)]{Guo}. Fig. 3 shows that  the BER performance of the codes with the PIC group decoding for $4$ transmit and $4$ receive antennas over Rayleigh block fading channels.  It can be seen that the codes $\mathbf{C}_{4,6,2}$ in \cite{ZhangGC09}, Guo-Xia's code and the code $\mathbf{\mathbf{B}}_{4,6,2}$ all obtain full diversity and have  very similar performance. It should be mentioned that the new  code $\mathbf{\mathbf{B}}_{4,6,2}$ only has half decoding complexity of the code  $\mathbf{C}_{4,6,2}$ in \cite{ZhangGC09}.  Guo-Xia's code is indeed a case of the systematic design of the new code in (\ref{eqn:new}). This can be seen from the equivalent channel matrix of Guo-Xia's code shown in \cite[Eq. (41)]{Guo} and that of the code  $\mathbf{B}_{4,6,2}$ in (\ref{eqn:H4}). Moreover, the code  $\mathbf{C}_{4,5,2}$ in \cite{ZhangGC09} has $1$ dB loss compared to the code $\mathbf{\mathbf{B}}_{4,6,2}$ due to its high bandwidth efficiency.
% In particular, the code $\mathbf{C}_{4,6,3}$   achieves the best BER performance among all the simulated codes, but it does not achieve full diversity in high SNR %\cite{ZhangGC09}.

\begin{figure}[t]
\centering
\includegraphics[width=3.4in]{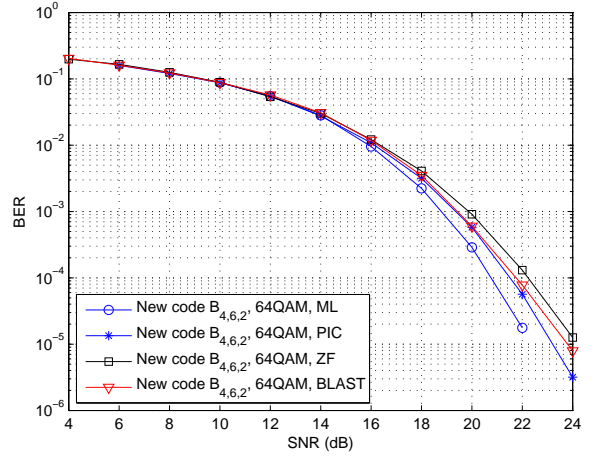}
\caption{Performance comparison of various detection methods (ZF, BLAST, ML and PIC) for the code $\mathbf{B}_{4,6,2}$  in a MIMO system with  $4$ transmit antennas and $4$ receive antennas at $8$ bps/Hz.} \label{fig:1}
\end{figure}

\begin{figure}[t]
\centering
\includegraphics[width=3.4in]{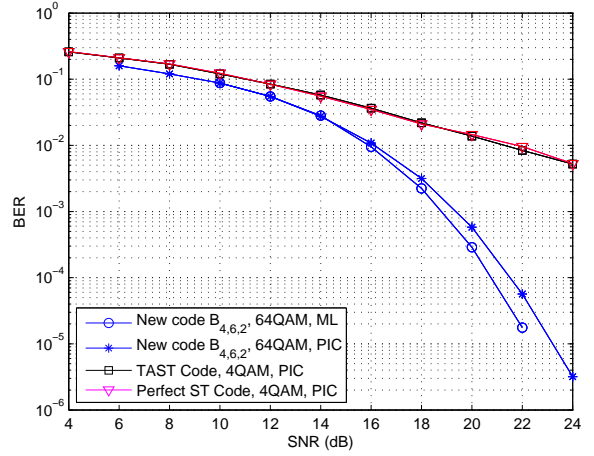}
\caption{Performance comparison among the proposed code $\mathbf{B}_{4,6,2}$, TAST code and perfect ST code for a MIMO system with  $4$ transmit antennas and $4$ receive antennas at $8$ bps/Hz.} \label{fig:2}
\end{figure}

\begin{figure}[t]
\centering
\includegraphics[width=3.4in]{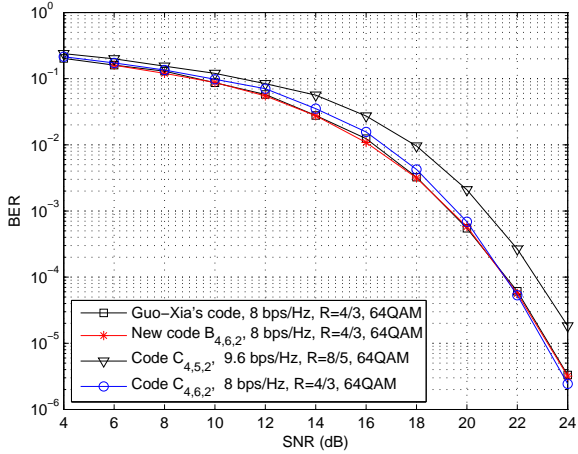}
\caption{BER performance of various codes with PIC group decoding for a MIMO system with  4 transmit antennas and 4 receive antennas.} \label{fig:ISIT2010}
\end{figure}

 \section{Conclusion}

In this paper, a design of full-diversity  STBC with reduced-complexity PIC group decoding was proposed. The proposed code design can obtain full diversity under  both ML decoding and PIC group decoding. Moreover, the decoding complexity of the full diversity STBC is equivalent to a joint ML decoding of $M/2$ symbols for $M$ transmit antennas while the rate can be up to $2$ symbols per channel use. For example, in a MIMO system with $4$ transmit antennas the full diversity can be achieved by the proposed codes with double-symbol ML decoding complexity and a code rate of $4/3$. Simulation results were shown to validate that the proposed codes achieve the full diversity gain with a low complexity decoding.   Although in this paper only Alamouti code is used in our new design, the method can be generalized to a general OSTBC.

\section*{APPENDIX   - Proof of Lemma 1}

%\subsection{Proof of Lemma 1.1}

 For MISO,  the received signal is given by
%\begin{equation}
$\mathbf{Y}= \sqrt{\rho/\mu}\mathbf{\mathbf{B}}\mathbf{h}+\mathbf{W}$.
%\end{equation}
Considering the code structure in (\ref{eqn:new}), we can rewrite $\mathbf{Y}$ in the matrix form as follows,
\begin{equation}
\left[\begin{array}{c}
\mathbf{y}_{1}\\
\mathbf{y}_{2}
\end{array}
\right]=
\sqrt{\frac{\rho}{\mu}}
\left[\begin{array}{cc}
\mathbf{C}^{1} & \mathbf{C}^{2}\\
-\left(\mathbf{C}^{2}\right)^\ast & \left(\mathbf{C^{1}}\right)^\ast\\
\end{array}\right]
\left[\begin{array}{c}
\mathbf{h}_{1}\\
\mathbf{h}_{2}
\end{array}\right]
+\left[\begin{array}{c}
\mathbf{w}_{1}\\
\mathbf{w}_{2}
\end{array}
\right],
\end{equation}
where $\mathbf{y}_1,\mathbf{y}_2\in \mathbb{C}^{T/2 \times 1}$, $\mathbf{w}_1,\mathbf{w}_2\in \mathbb{C}^{T/2 \times 1}$ and $\mathbf{h}_1,\mathbf{h}_2\in \mathbb{C}^{M/2 \times 1}$. We further have
\begin{equation}
\begin{split}
\mathbf{y}_1=& \sqrt{\rho/\mu}\left(\mathbf{C}^1\mathbf{h}_1+\mathbf{C}^2\mathbf{h}_2\right)+\mathbf{w}_1,\\
\mathbf{y}_2=& \sqrt{\rho/\mu}\left(-(\mathbf{C}^2)^\ast\mathbf{h}_1+(\mathbf{C}^1)^\ast\mathbf{h}_2\right)+\mathbf{w}_2.
\end{split}
\end{equation}

With the expansion of code matrix $\mathbf{C}^1$ and $\mathbf{C}^2$ shown in (\ref{eqn:wei}), we can rewrite $\mathbf{y}_1$ and $\mathbf{y}_2$ as
\begin{equation}
\begin{split}
\mathbf{y}_1=& \sqrt{\rho/\mu}\left(\sum^P_{p=1}\mathbf{C}^1_p\mathbf{h}_1+\sum^P_{p=1}\mathbf{C}^2_p\mathbf{h}_2\right)+\mathbf{w}_1,\\
\mathbf{y}_2=& \sqrt{\rho/\mu}\left(\sum^P_{p=1}-(\mathbf{C}^2_p)^\ast\mathbf{h}_1+\sum^P_{p=1}(\mathbf{C}^1_p)^\ast\mathbf{h}_2\right)+\mathbf{w}_2,
\end{split}
\end{equation}
where
\begin{eqnarray}
\mathbf{C}^i_p=&\left[\begin{array}{c}
      \mathbf{0}_{(p-1)\times (M/2)}\\
      \mathrm{diag}(\mathbf{X}^i_p)\\
      \mathbf{0}_{(P-p)\times (M/2)}
      \end{array}
  \right],\,\, p=1,2,\cdots,P;\,\, i=1,2.
  \end{eqnarray}

Equivalently, we have
\begin{equation}
\begin{split}
\mathbf{y}_1=& \sqrt{\rho/\mu}\left(\sum^P_{p=1}\mathcal{H}_{1,p}\mathbf{X}^1_p+\sum^P_{p=1}\mathcal{H}_{2,p}\mathbf{X}^2_p\right)+\mathbf{w}_1,\\
-(\mathbf{y}_2)^\ast=& \sqrt{\rho/\mu}\left(\sum^P_{p=1}(\mathcal{H}_{1,p})^\ast\mathbf{X}^2_p+\sum^P_{p=1}-(\mathcal{H}_{2,p})^\ast\mathbf{X}^1_p\right)+(-\mathbf{w}_2)^\ast,
\end{split}
\end{equation}
where
\begin{eqnarray}
\mathcal{H}_{i,p}= \left[\begin{array}{c}
      \mathbf{0}_{(p-1)\times (M/2)}\\
      \mathrm{diag}(\mathbf{h}_i)\\
      \mathbf{0}_{(P-p)\times (M/2)}
      \end{array}
  \right],\,\, p=1,2,\cdots,P;\,\, i=1,2.
  %%%\\
%%%  \mathcal{H}_{2,p}=&\left[\begin{array}{c}
%%%      \mathbf{0}_{(p-1)\times (M/2)}\\
%%%      \mathrm{diag}(\mathbf{h}_2)\\
%%%      \mathbf{0}_{(P-p)\times (M/2)}
%%%      \end{array}
%%%  \right], p=1,2,\cdots,P.
%%%  \end{split}
  \end{eqnarray}

Using   $\mathbf{X}^i_p=\mathbf{\Theta}\mathbf{s}^i_p$ shown in (\ref{eqn:RM}), we then have
\begin{equation}
\begin{split}
\mathbf{y}_1=& \sqrt{\rho/\mu}\left(\sum^P_{p=1}\mathcal{H}_{1,p}\mathbf{\Theta}\mathbf{s}^1_p+\sum^P_{p=1}\mathcal{H}_{2,p}\mathbf{\Theta}\mathbf{s}^2_p\right)+\mathbf{w}_1\\
=&\sqrt{\rho/\mu}\left(\mathcal{H}_1\mathbf{s}^1+\mathcal{H}_2\mathbf{s}^2\right)+\mathbf{w}_1,\\
-(\mathbf{y}_2)^\ast=& \sqrt{\rho/\mu}\left(\sum^P_{p=1}(\mathcal{H}_{1,p})^\ast\mathbf{\Theta}\mathbf{s}^2_p+\sum^P_{p=1}(-\mathcal{H}_{2,p})^\ast\mathbf{\Theta}\mathbf{s}^1_p\right)-(\mathbf{w}_2)^\ast\\
=&\sqrt{\rho/\mu}\left((\mathcal{H}_1)^\ast\mathbf{s}^2+(-\mathcal{H}_2)^\ast\mathbf{s}^1\right)-(\mathbf{w}_2)^\ast.
\end{split}
\end{equation}
Therefore, %$\mathbf{Y}^{'}$ can be written as
\begin{eqnarray}
  \left[\begin{array}{c}
\mathbf{y}_{1}\\
-(\mathbf{y}_{2})^\ast
\end{array}
\right] &
=&\sqrt{\frac{\rho}{\mu}}
\left[\begin{array}{cc}
\mathcal{H}_{1} & \mathcal{H}_{2}\\
-(\mathcal{H}_{2})^\ast & (\mathcal{H}_{1})^\ast
\end{array}\right]
\left[\begin{array}{c}
\mathbf{s}^{1}\\
\mathbf{s}^{2}
\end{array}\right]
+\left[\begin{array}{c}
\mathbf{w}_{1}\\
-\mathbf{w}_{2}^*
\end{array}
\right]\nonumber
\\
&=&\sqrt{{\rho}/{\mu}}\mathcal{H}\mathbf{s}+\mathbf{w}^{'},
\end{eqnarray}
where the equivalent channel matrix $\mathcal{H} \in \mathbf{C}^{T \times MP}$ is given by
 \begin{eqnarray}
 \left[\begin{array}{cccccccc}
                     \mathcal{G}^1_1 & \mathcal{G}^1_2 & \cdots & \mathcal{G}^1_P & \mathcal{G}^2_1 & \mathcal{G}^2_2 & \ldots & \mathcal{G}^2_P \\
                      -(\mathcal{G}^2_1)^\ast & -(\mathcal{G}^2_2)^\ast & \cdots & -(\mathcal{G}^2_P)^\ast & (\mathcal{G}^1_1)^\ast & (\mathcal{G}^1_2)^\ast & \cdots & (\mathcal{G}^1_P)^\ast
                  \end{array}
                 \right] ,\nonumber
 \end{eqnarray}
with
  \begin{eqnarray}
 \mathcal{G}^i_p = \left[\begin{array}{c}
                     \mathbf{0}_{(p-1)\times M/2} \\
                     \mathrm{diag}(\mathbf{h}_i)\mathbf{\Theta} \\
                      \mathbf{0}_{(P-p)\times M/2}
                  \end{array}
                 \right], i=1,2;\,\,\,p=1,2,\cdots,P.
 \end{eqnarray}


\begin{thebibliography}{1}

{

\bibitem{Telatar} E. Telatar, ``Capacity of multi-antenna Gaussian channels,''  {\em Europ. Trans. Telecommun.}, vol. 10, pp. 585--595, Nov. 1999.

 \bibitem{Naguib} A. F. Naguib, N. Seshadri, and A. R. Calderbank, ``Increasing data rate over wireless channels,'' {\em IEEE Signal Processing Mag.}, vol. 48, pp. 76--92, May 2000.

 \bibitem{ZhangMag} W. Zhang, X.-G. Xia, and K. B. Letaief, ``Space-time/frequency coding for
 MIMO-OFDM in next generation broadband wireless systems,''   {\em  IEEE Wireless Commun. Mag.}, vol. 14, no. 3, pp. 32--43, June 2007.


\bibitem{Alamouti} S. M. Alamouti, ``A simple transmit diversity technique for wireless communication,''
{\em IEEE J. Sel. Areas Commun.}, vol. 16, pp. 1451--1458, Oct. 1998.
 % \vspace{-0.1cm}

\bibitem{Tarokh98} V. Tarokh, N. Seshadri, and A. Calderbank, ``Space-time codes for high data
rate wireless communications: Performance criterion and code construction,''
{\em IEEE Trans. Inf. Theory}, vol. 44, pp. 744--765, Mar. 1998.

\bibitem{Tarokh00} V. Tarokh, H. Jafarkhani, and A. R. Calderbank, ``Space-time block codes from orthogonal designs,'' {\em
IEEE Trans. Inf. Theory}, vol. 45, pp. 1456--1467, July 1999. Also, ``Corrections to
 `Space-time block codes from orthogonal designs','' {\em IEEE Trans. Inf. Theory}, vol.
 46, p. 314, Jan. 2000.





 % \vspace{-0.1cm}


%\bibitem{Liang} X. Liang, ``Orthogonal designs with maximal rates,'' {\em IEEE Trans. Inf.
%Theory}, vol. 49, pp. 2468--2503, Oct. 2003.
% \vspace{-0.1cm}

%\bibitem{SXL} W. Su, X.-G. Xia, and K. J. R. Liu, ``A systematic design of high-rate complex orthogonal
%space-time block codes,'' {\em IEEE Commun. Lett.}, vol. 8, no. 6, pp.380--382, June
%2004.
% \vspace{-0.1cm}

 \bibitem{Lu} K. Lu, S. Fu, and X.-G. Xia, ``Closed form designs of complex orthogonal space-time block codes of rates
 $(k+1)/(2k)$ for $2k-1$ or $2k$ transmit antennas,''   {\em IEEE Trans. Inf. Theory},
 vol. 51, pp. 4340--4347, Dec. 2005.
  %\vspace{-0.1cm}

\bibitem{Wang} H. Wang and X.-G Xia, ``Upper bounds of rates of complex orthogonal space-time block codes,'' {\em IEEE Trans. Inf. Theory}, vol 49, pp. 2788--2796, Oct. 2003.
 % \vspace{-0.1cm}

%\bibitem{Hassibi} B. Hassibi  and B. M. Hochwald, ``High-rate codes that are linear in space and time,'' {\em IEEE
%Trans. Inf. Theory}, vol. 48, pp. 1804--1824, July 2002.
 % \vspace{-0.1cm}

%\bibitem{Damen} M. O. Damen, A. Tewfik, and J. C. Belfiore, ``A construction of a space-time code based on
%number theory,'' {\em IEEE Trans. Inf. Theory}, vol. 48, pp. 753--760, Mar. 2002.
 % \vspace{-0.1cm}

\bibitem{Jafar} H. Jafarkhani, ``A quasi-orthogonal space-time block code,'' {\em IEEE Trans. Commun.}, vol. 49, pp. 1--4, Jan. 2001.

\bibitem{Tirkkonen}  O. Tirkkonen, A. Boariu, and A. Hottinen, ``Minimal non-orthogonality rate $1$ space-time block code for $3+$ Tx antennas,'' in {\em Proc. IEEE 6th Int. Symp. on Spread-Spectrum Tech. and Appl.
(ISSSTA 2000)}, Sep. 2000, pp. 429--432.

\bibitem{Papadias}C. B. Papadias and G. J. Foschini, ``Capacity-approaching space-time codes for systems employing four transmit antennas,'' {\em IEEE Trans. Inf. Theory}, vol. 49, pp. 726--733, Mar. 2003.


\bibitem{Tirkkonen2}  O. Tirkkonen, ``Optimizing space-time block codes by constellation rotations,'' {\em Proc. Finnish Wireless Commun. Workshop}, Finland, pp. 59--60, Oct. 2001.

\bibitem{Sharma} N. Sharma and C. B. Papadias, ``Full-rate full-diversity linear quasi-orthogonal space-time codes for any number of transmit antennas,'' {\em EURASIP J. Applied Signal Processing}, vol. 9, pp. 1246--1256, Aug. 2004.

\bibitem{SuXia} W. Su and X.-G. Xia, ``Signal constellations for quasi-orthogonal
space-time block codes with full diversity,'' {\em IEEE Trans. Inf. Theory},
vol. 50, pp. 2331--2347, Oct. 2004.





 \bibitem{Khan}  Z. A. Khan and B. S. Rajan, ``Single-symbol maximum-likelihood decodable linear STBCs,'' {\em IEEE Trans. Inf. Theory}, vol. 52, pp. 2062--2091, May 2006.

\bibitem{Yuen} C. Yuen, Y. Guan, and T. T. Tjhung, ``Quasi-orthogonal STBC with minimum decoding complexity,'' {\em IEEE Trans. Wireless Commun.}, vol. 4, pp. 2089--2094, Sep. 2005.


 \bibitem{HWang} H. Wang, D. Wang, and X.-G. Xia, ``On optimal quasi-orthogonal space-time block codes with minimum decoding complexity,'' {\em IEEE Trans. Inf. Theory}, vol. 55,  pp. 1104--1130, Mar. 2009.



\bibitem{Dao} D. N. Dao, C. Yuen, C. Tellambura,   Y. L. Guan, and T. T. Tjhung,  	``Four-group decodable space-time block codes,''
 {\em IEEE Trans. Signal Processing}, vol. 56, pp. 424--430, Jan. 2008.



\bibitem{Karmakar} S. Karmakar and B. S. Rajan, ``Multigroup decodable STBCs from clifford algebras,'' {\em IEEE Trans. Inf. Theory}, vol. 55, pp. 223--231, Jan. 2009.





 \bibitem{Damen} M. O. Damen, A. Tewfik, and J. C. Belfiore, ``A construction of a space-time code based on
 number theory,'' {\em IEEE Trans. Inf. Theory}, vol. 48, pp. 753--760, Mar. 2002.

 %\bibitem{oggier0} F. Oggier and E. Viterbo, ``Algebraic number theory and code design for Rayleigh fading channels,'' {\em Foundations and Trends in Communications and % %Information Theory}, vol. 1, no. 3, pp. 333--415.


\bibitem{gamal}
H. El Gamal and M. O. Damen, ``Universal space-time coding,'' {\em IEEE Trans. Inf. Theory,} vol. 49, pp. 1097--1119, May 2003.

\bibitem{WangGY} G. Wang, H. Liao, H. Wang,  and X.-G. Xia, ``Systematic and optimal cyclotomic lattices and diagonal space-time block code designs,''
{\em IEEE Trans. Inf. Theory}, vol. 50, pp. 3348--3360, Dec. 2004.
%\small

\bibitem{WangGY2} G. Wang and X.-G. Xia, ``On optimal multilayer cyclotomoc space-time code designs,'' {\em IEEE Trans. Inf. Theory}, vol. 51, pp. 1102--1135, Mar. 2005.



\bibitem{rajan0}
B. A. Sethuraman, B. S. Rajan, and V. Shashidhar, ``Full-diversity, high-rate
  space-time block codes from division algebras,'' {\em IEEE Trans. Inf.
  Theory}, vol. 49,  pp. 2596--2616, Oct. 2003.





\bibitem{rajan1}
 Kiran T.  and B. S. Rajan, ``{STBC}-scheme with nonvanishing determinant for
  certain number of transmit antennas,''  {\em IEEE Trans. Inf. Theory},
  vol. 51,   pp. 2984--2992, Aug. 2005.


\bibitem{kumar1}
P. Elia, K. R. Kumar, S. A. Pawar, P. V. Kumar, and H.-F. Lu, ``Explicit
  space-time codes achieving the diversity-multiplexing gain tradeoff,''
  \emph{IEEE Trans. Inf. Theory}, vol. 52,   pp. 3869--3884, Sep. 2006.

 \bibitem{oggier} F. Oggier, G. Rekaya, J.-C. Belfiore, and E. Viterbo, ``Perfect space-time block codes,''
 {\em  IEEE Trans. Inf. Theory}, vol. 52, pp. 3885--3902, Sep. 2006.

\bibitem{kumar2} P. Elia, B. A. Sethuraman, and P. V. Kumar, ``Perfect
space-time codes for any number of antennas,'' {\em IEEE Trans. Inf. Theory}, vol. 53,  pp. 3853--3868, Nov. 2007.



\bibitem{oggier1}  F. Oggier, J.-C. Belfiore, and E. Viterbo, ``Cyclic division algebras: A tool for space-time coding,''
{\em Foundations and Trends in Communications and Information Theory}, vol. 4, no. 1, pp. 1--95.


\bibitem{Guo09} X. Guo and X.-G. Xia, ``An elementary condition for non-norm elements,'' {\em IEEE Trans. Inf. Theory}, vol. 55, pp. 1080-1085, Mar. 2009.


% \bibitem{ZhangTCOM} W. Zhang, X.-G. Xia, and P. C. Ching, ``High-rate full-diversity space-time-frequency
%codes for broadband MIMO blocking-fading channels,'' {\em IEEE
%Trans. Commun.}, vol. 55, no. 1, pp. 25--34,  Jan. 2007.



%\bibitem{JaldenSP} J. Jalden and B. Ottersten, ``On the complexity of sphere decoding in digital communications,''
%{\em IEEE Trans. Signal Process.}, vol. 53, pp. 1474--1484, Apr. 2005.




 \bibitem{Choi} X. H. Nguyen and J. Choi, ``Joint design of groupwise STBC and SIC based receiver,'' {\em IEEE Commun. Lett.}, vol. 12,
 pp. 115--117, Feb. 2008.

 \bibitem{Hong} E. Biglieri, Y. Hong, and E. Viterbo, ``On fast-decodable space-time block codes,'' {\em IEEE Trans. Inf. Theory}, vol. 55, pp. 524--530, Feb. 2009.

\bibitem{Calderbank}
P. Rabiei, N. Al-Dhahir, and R. Calderbank, ``New rate-2 STBC design for 2 TX with reduced-complexity maximum likelihood decoding,'' {\em IEEE Trans. Wireless Commun.}, vol. 8,  pp. 1803--1813, Apr. 2009.

% \bibitem{ZhangTCOM} W. Zhang, X.-G. Xia, and P. C. Ching, ``High-rate full-diversity space-time-frequency
%codes for broadband MIMO blocking-fading channels,'' {\em IEEE
%Trans. Commun.}, vol. 55, no. 1, pp. 25--34,  Jan. 2007.
 % \vspace{-0.1cm}




%\bibitem{Tepede}
%C. Tepedelenlioglu, ``Maximum multipath diversity with linear equalization in precoded OFDM systems,'' {\em IEEE Trans. Inf. Theory}, vol. 50, pp. 232--235, Jan. 2004.

\bibitem{Aria} A. Hedayat and A. Nosratinia, ``Outage and diversity of linear receivers in flat-fading MIMO channels,'' {\em IEEE Trans. Signal Processing}, vol. 55, pp. 5868--5873, Dec. 2007.

%\bibitem{Ma} X. Ma and W. Zhang, ``Fundamental limits of linear equalizers: diversity, capacity, and complexity,'' {\em IEEE Trans. Inf. Theory}, vol. 54, pp. 3442--3456, %Aug. 2008.

\bibitem{Caire} K. Raj Kumar, G. Caire, and A. L. Moustakas, ``Asymptotic performance of linear receivers in MIMO fading channels,'' {\em IEEE Trans. Inf. Theory}, vol. 55, pp. 4398--4418, Oct. 2009.


\bibitem{Tse} L. H. Grokop and D. N. C. Tse, ``Diversity-multiplexing tradeoff in ISI channels,'' {\em IEEE Trans. Inf. Theory}, vol. 55, pp. 109--135, Jan. 2009.

\bibitem{Aria2} A. Tajer and A. Nosratinia, ``Diversity order in ISI channels with single carrier frequency domain equalizers,''  {\em IEEE Trans. Wireless Commun.}, vol. 9, pp. 1022--1032, Mar. 2010.

 \bibitem{Jiang}
 Y. Jiang, M. K. Varanasi, and J. Li, ``Performance analysis of ZF and MMSE equalizers for MIMO systems: An in-depth study of the high SNR regime,''  {\em IEEE Trans. Inf. Theory}, to appear. \url{http://www.sal.ufl.edu/yjiang/papers/VbitwcR6.pdf}.




\bibitem{Liu} J. Liu, J.-K. Zhang, and K. M. Wong, ``Full-diversity codes for MISO systems equipped with linear or ML detectors,''
{\em IEEE Trans. Inf. Theory}, vol. 54,   pp. 4511--4527, Oct. 2008.
  %\vspace{-0.1cm}

\bibitem{Shang} Y. Shang and X.-G. Xia, ``Space-time block codes achieving full diversity with linear receivers,''
{\em IEEE Trans. Inf. Theory}, vol. 54,   pp. 4528--4547, Oct. 2008.
 % \vspace{-0.1cm}


\bibitem{ZhangICASSP09} W. Zhang and J. Yuan, ``A simple design of space-time block codes achieving full diversity with linear receivers,'' in {\em Proc. IEEE ICASSP},  Taipei, Apr. 20-24, 2009, pp. 2729--2732.


\bibitem{ZhangISIT09} W. Zhang and J. Yuan, ``Linear receiver based high-rate space-time block codes,'' in {\em Proc. IEEE ISIT},  Seoul, Korea, June 29--July 3, 2009, pp. 94--98.



\bibitem{WangHM} H. Wang, X.-G. Xia, Q. Yin, and B. Li, ``A family of space-time block codes achieving full diversity with linear receivers,'' {\em IEEE Trans. Commun.}, vol. 57, pp. 3607--3617, Dec. 2009.



\bibitem{Guo} X. Guo and X.-G. Xia, ``On full diversity space-time block codes with partial interference cancellation group decoding,'' {\em IEEE Trans. Inf. Theory}, vol. 55, pp. 4366--4385, Oct. 2009. Also, ``Corrections to `On full diversity space-time block codes with partial interference cancellation group decoding','' \url{http://www.ece.udel.edu/~xxia/correction_guo_xia.pdf}.



\bibitem{ZhangGC09}   W. Zhang, T. Xu, and X.-G. Xia, ``Two designs of  space-time block codes achieving full diversity with partial interference cancellation group decoding,'' {\em IEEE Trans. Inf. Theory}, submitted. Also, 
     \url{http://arxiv.org/abs/0904.1812v3}

\bibitem{Foschini} G. J. Foschini, ``Layered space-time architecture for wireless communication in a fading
environment when using multi-element antennas,'' {\em Bell Labs Tech. J.}, vol. 1, no. 2, pp.
41--59, 1996.

%\bibitem{Viterbo} E. Viterbo and J. Boutros, ``A universal lattice code decoder for fading channels,''
%{\em IEEE Trans. Inf. Theory}, vol. 45, pp. 1639--1642, July 1999.

}
\end{thebibliography}
\end{document}